\documentclass{article}

\usepackage{algorithm}
\usepackage[noend]{algpseudocode}
\usepackage{amssymb}
\usepackage[absolute]{textpos}
\usepackage[all=normal,bibliography=tight]{savetrees}
\usepackage{fullpage}
\usepackage{amsmath}
\usepackage{amstext,amsfonts,amsthm,amsmath,amssymb}

\usepackage{graphicx,tikz}
\usetikzlibrary{cd}
\usetikzlibrary{arrows,shapes}

\newtheorem{lemma}{Lemma}[section]

\newtheorem{corollary}[lemma]{Corollary}
\newtheorem{theorem}[lemma]{Theorem}

\newtheorem{fact}[lemma]{Fact}
\newtheorem{conjecture}[lemma]{Conjecture}

\newcommand{\mad}{\mathrm{mad}}
\newcommand{\ora}{\overrightarrow}

\begin{document}

\title{Decreasing the maximum average degree by deleting an independent set or a $d$-degenerate subgraph\thanks{This research is a part of projects that have received funding from the European Research Council (ERC)
under the European Union's Horizon 2020 research and innovation programme
Grant Agreement 714704 (W. Nadara) and 677651 (M. Smulewicz).}}
\date{\today}
\author{Wojciech Nadara\thanks{Institute of Informatics, University of Warsaw, Poland (\texttt{w.nadara@mimuw.edu.pl})} \ \
        Marcin Smulewicz\thanks{Institute of Informatics, University of Warsaw, Poland (\texttt{m.smulewicz@mimuw.edu.pl})}}

\maketitle

\begin{textblock}{20}(0, 13.0)
\includegraphics[width=40px]{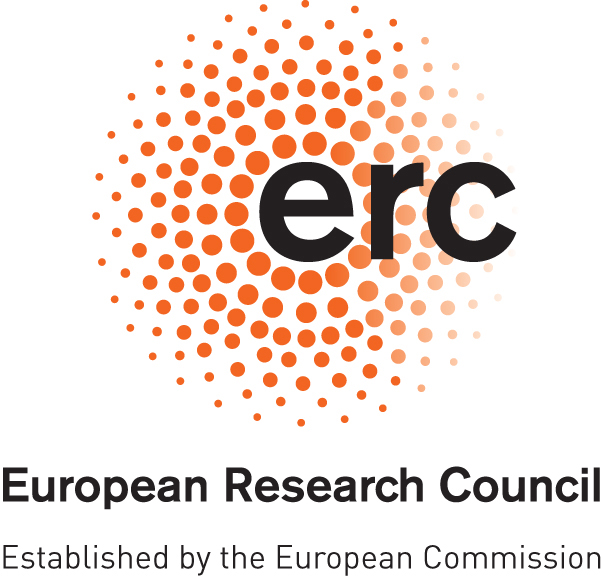}%
\end{textblock}
\begin{textblock}{20}(-0.25, 13.4)
\includegraphics[width=60px]{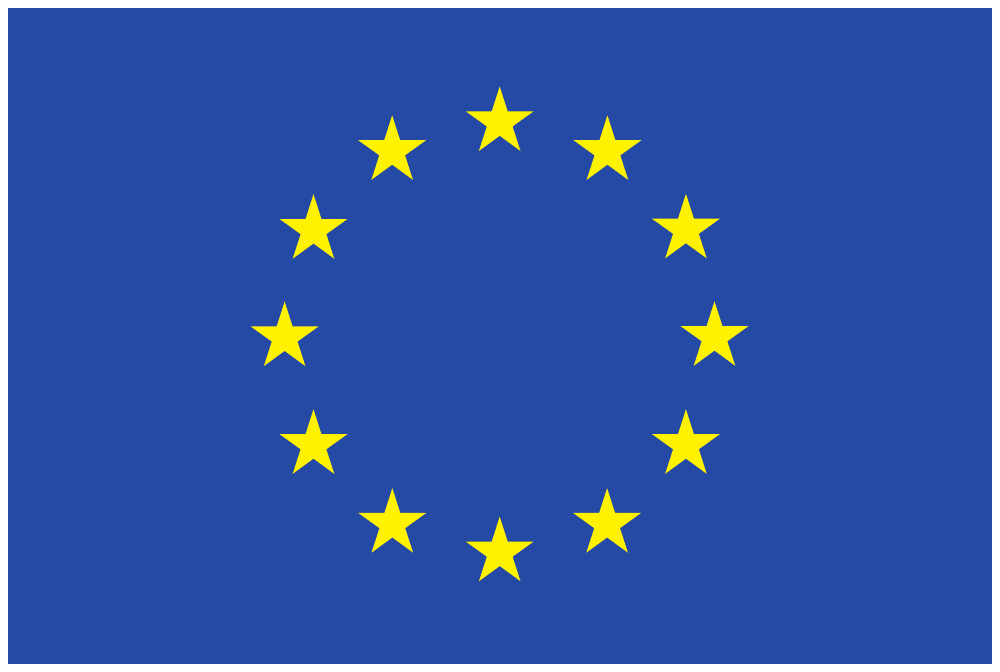}%
\end{textblock}

\begin{abstract}
The maximum average degree $\mad(G)$ of a graph $G$ 
is the maximum average degree over all subgraphs of $G$.
In this paper, we prove that for every $G$ and positive integer $k$
such that $\mad(G) \ge k$ there exists $S \subseteq V(G)$ such that
$\mad(G - S) \le \mad(G) - k$ and $G[S]$ is $(k-1)$-degenerate.
Moreover, such $S$ can be computed in polynomial time.
In particular, if $G$ contains at least one edge then there exists an independent set $I$ in $G$ such that
$\mad(G-I) \le \mad(G)-1$ and if $G$ contains a cycle then there exists an induced forest $F$ such that $\mad(G-F) \le \mad(G) - 2$.


\end{abstract}

\section{Introduction}
The maximum average degree (abbreviated as $\mad$) of a graph is a heavily studied notion.
 Multiple results show that a lower or upper bound on $\mad$ implies an existence of a particular partition of vertices of $G$, e.g., 
\cite{Borodin2011, Chen, Dross, Kopreski}.
Another class of results consider edge partitions, including \cite{Planar10, Planar9, Planar11}.
In these directions, particular attention has been paid to planar graphs, where, due to the inequality
$(\mad(G)-2)(g(G)-2) < 4$, an upper bound on $\mad$ can be inferred from a lower bound on the girth.

A graph parameter $f(G)$ will be called \emph{partitionable} if for every undirected simple graph $G$ and positive real numbers $a$ and $b$ such that $f(G) < a+b$, the vertex set $V(G)$ can be partitioned into $A$ and $B$ so that
$f(G[A]) < a$ and $f(G[B]) < b$. It is quite simple to prove that degeneracy, maximum degree and minimum degree are all partitionable parameters. However it is an open problem whether $\mad$ is partitionable. Such result would agree exactly with or even improve many existing results, for example mentioned in~\cite{Dross, Kopreski}. We answer this question positively for cases $a=1$ and $a=2$. It is a consequence of a following theorem which is the main result of this paper.

\begin{theorem} \label{mainres}
 For every undirected simple graph $G$ and a positive integer $k$ such that $\mad(G) \geq k$
 there exists $S \subseteq V(G)$ such that $G[S]$ is $(k-1)$-degenerate and $\mad(G - S) \le \mad(G) - k$.
 Moreover such $S$ can be computed in polynomial time.
\end{theorem}

Up to our knowledge this is the first theorem of a kind where we are given a graph with
bounded value of its $\mad$ where we partition its vertex set into some
parts so that their values of $\mad$ are smaller, however they need not be bounded 
by absolute constant. This is opposed to all results where every resulting part
induces a forest or is an independent set or has maximum degree $1$, etc.

Our results can be applied as a tool for directly deriving many results for some specific
sparse graph classes, for example planar graphs with constraints on girth.
It seems that our results do not show as much expressive power as it is possible
to get on such restrictive graph classes
(where arguments specifically adjusted to the researched restricted graph classes can be used),
which is a price for deriving them from a more general theorem.
However, our results can be seen as a nice way of unifying
these results and there are cases where using our results improves the state of the art.

Our results imply a positive answer for the
open problem presented in \cite{EDU} (problem 2 from final remarks)
what in turn implies sub-exponential bound on the diameter of reconfiguration graphs of
colourings for graphs with any bounded maximum average degree.

The rest of the paper is organized as follows.
In Section \ref{Prelims}, we introduce a few useful notions.
In Section \ref{main}, we state our main result and provide its proof.
In Section \ref{conclusions}, we present some direct consequences of our results and conclude this paper.

\section{Preliminaries} \label{Prelims}

Theorems proved in this paper will be about simple undirected graphs, however
multiple directed graphs will show up throughout the proofs.

Undirected edge between vertices $u$ and $v$ will be denoted as $uv$.
Directed edge from $u$ to $v$ will be denoted as $\ora{uv}$.

If $G$ is a graph and $A$ is a subset of its vertices then by $G[A]$ we denote subgraph
of $G$ induced on vertices of $A$. Length of the shortest cycle in a graph $G$ is called girth and will be denoted as $g(G)$. If $G$ is a forest we set that $g(G) = \infty$.
$\Delta(G)$ denotes maximum degree of a vertex.


The maximum average degree of a given graph $G$ is defined as follows:
$$\mad(G) := \max_{H \subseteq G, H \neq \emptyset} \frac{2|E(H)|}{|V(H)|}.$$
We assume that $\mad$ of a graph with an empty vertex set is $-\infty$.

We say that undirected graph $G$ is $k$-degenerate if each of its subgraphs contains a vertex of degree at most~$k$.
Degeneracy of a graph is the smallest value of $k$ such that this graph is $k$-degenerate.

Let us note that class of $0$-degenerate graphs is exactly the same class of graphs as graphs with ${\mad(G) < 1}$, because both are just edgeless graphs. Moreover class of $1$-degenerate graphs is
exactly the same class of graphs as graphs with $\mad(G) < 2$, because both are just forests.

\section{Proof of Theorem \ref{mainres}} \label{main}

In order to prove Theorem \ref{mainres} we are going to investigate flow network
that allows to determine value of $\mad$ in polynomial time.
Example of such network can be found in \cite{FlowNetwork}, however we are going
to use one adjusted to our own use.

\begin{figure}[b]
\centering
\begin{minipage}{.2\textwidth}
\begin{tikzpicture}[->,>=stealth',shorten >=1pt,auto,node distance=1.6cm,
        thick,main node/.style={circle,draw,minimum size=0.6cm,inner sep=0pt]}]

  \node[main node] (1) {$x$};
  \node[main node] (2) [right of=1] {$z$};
  \node[main node] (3) [below right of=1] {$y$};
  \node[main node] (4) [below of=3] {$u$};  
  \path[-]
  (1) edge node {} (2)
  (1) edge node {} (3)
  (2) edge node {} (3)
  (3) edge node {} (4);
\end{tikzpicture}

\end{minipage}
\begin{minipage}{.35\textwidth}

\begin{tikzpicture}[->,>=stealth',shorten >=1pt,auto,node distance=1.6cm,
        thick,main node/.style={circle,draw,minimum size=0.6cm,inner sep=0pt]}]

  \node[main node] (1) {$s$};
  \node[main node] (3) [below of=1] {$v_{xz}$};
  \node[main node] (2) [left of=3] {$v_{xy}$};
  \node[main node] (4) [right of=3] {$v_{yz}$};
  \node[main node] (5) [right of=4] {$v_{yu}$};
  
  \node[main node] (6) [below of=2]  {$x$};
  \node[main node] (7) [below of=3]  {$y$};
  \node[main node] (8) [below of=4]  {$z$};
  \node[main node] (9) [below of=5]  {$u$};
  
  \node[main node] (10) [below of=7]  {$t$};

  \path[->]
  (1) edge [left] node {$1$} (2)
  (1) edge [left] node {$1$} (3)
  (1) edge [right] node {$1$} (4)
  (1) edge [right] node {$1$} (5)
  (2) edge [left] node {$\infty$} (6)
  (2) edge [above, near start] node {$\infty$} (7)

  (3) edge [above, near start] node {$\infty$} (6)
  (3) edge [above, near start] node {$\infty$} (8)

  (4) edge [above, near start] node {$\infty$} (7)
  (4) edge [near end] node {$\infty$} (8)
  
  (5) edge [above, near start] node {$\infty$} (7)
  (5) edge [right] node {$\infty$} (9)

  (6) edge [left] node {$c$} (10)
  (7) edge [left] node {$c$} (10)
  (8) edge [right] node {$c$} (10)
  (9) edge [right] node {$c$} (10);
\end{tikzpicture}

\end{minipage}

\centering
\caption{Example of graph $G$ and flow network $F(G,c)$ corresponding to it.}
\label{network}
\end{figure}
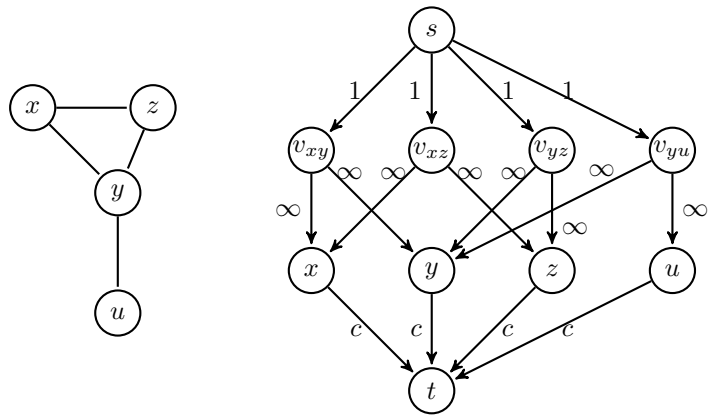

Let us define a flow network $F(G,c)$ for given undirected graph $G$
and any nonnegative real number $c$.
The network will consist of 
one node for each $v \in V(G)$,
one node for each $e \in E(G)$ denoted as $v_e$
and two special nodes $s$ and $t$, respectively source and sink.
There will be three layers of directed edges in $F(G,c)$:
\begin{itemize}
\item The first layer -- Edges of capacity one from $s$ to each node $v_e$.
\item The second layer -- Edges of infinite capacity from each $v_e$ where $e=uw \in E(G)$ to $u$ and to $w$.
\item The third layer -- Edges of capacity $c$ from each $v \in V(G)$ to $t$.
\end{itemize}

\begin{lemma}\label{flow}
  For any graph $G$ and any real number $c$, maximum flow between $s$ and $t$ in $F(G,c)$ is equal to $|E(G)|$
  if and only if $2c \ge \mad(G)$.
\end{lemma}

\begin{proof}
By max-flow min-cut theorem we know that maximum flow in a graph $G$
is equal to the minimum cut, so we are going to investigate structure of $s-t$ cuts
in this graph. We refer to cuts as sets of edges. Since edges in second layer have infinite capacity they surely do not belong to
any minimum cut. If no edge from third layer belongs to the cut then all edges
from first layer must belong to it and this is a cut of weight $|E(G)|$, so 
if maximum flow is smaller than $|E(G)|$ then there exists a minimum cut with some edges
in third layer. Let us fix some minimal cut $C \subseteq E(F(G, c))$ and
let us assume that $ W \subseteq V(G)$ is a nonempty subset of all vertices of $V(G)$
such that edges $\ora{wt}$ for all $w \in W$ belong to $C$. Let $H = G[W]$.
If $e \not\in E(H)$ then $\ora{sv_e}$ has to belong to $C$.
Observe that all mentioned edges, that is $\ora{wt}$ for $w \in W$ and $\ora{sv_e}$ for $e \not\in E(H)$ already form an $s-t$ cut.
Its weight is $c|V(H)| + |E(G)| - |E(H)|$. If this value
is less than $|E(G)|$ then we know that maximum flow in this graph is less than $|E(G)|$.
However if for any $H$ this value is not smaller than $|E(G)|$ then we know that
maxflow in this graph is $|E(G)|$.

We get that maxflow in this graph is smaller than $|E(G)|$ if and only if
there exists $H \subseteq G$ such that 
$c|V(H)| + |E(G)| - |E(H)| < |E(G)| \Leftrightarrow c|V(H)| < |E(H)| \Leftrightarrow c < \frac{|E(H)|}{|V(H)|}$.
Maximum value of $\frac{|E(H)|}{|V(H)|}$ equals $\frac{\mad(G)}{2}$, so we get that
maxflow in $F(G, c)$ is equal to $|E(G)|$ is and only if $c \ge \frac{\mad(G)}{2}$, as desired. 

\end{proof}

Let us note that by using Lemma \ref{flow}, observing that $\mad(G) = \frac{a}{b}$ for some
$a, b \in \mathbb{Z}$ and $a \le n^2, b \le n$ and knowing that we can compute maximum flow
in polynomial time, we can conclude
that $\mad(G)$ can be computed in polynomial time.

Let us fix any graph $G$ and denote $F := F(G,\frac{\mad(G)}{2})$.
Let us define a directed graph $G_f$ for a given $s-t$ flow $f$ in $F$ of capacity $|E(G)|$ by directing some of edges from $G$ and discarding the rest.
Flow $f$ routes one unit of flow through each $v_{uw}$.
Node $v_{uw}$ has two outgoing edges to $u$ and to $w$.
If $f$ sends more than $\frac{1}{2}$ unit of flow to $w$ then in $G_f$ we put
directed edge $\ora{uw}$,
similarly if $f$ sends more than $\frac{1}{2}$ unit of flow to $u$ we put
edge $\ora{wu}$. Otherwise if $f$ sends exactly $\frac{1}{2}$ unit to both $u$ and $w$
we simply discard this edge.

\begin{lemma}\label{acyclic}
  There exists flow $f$ of capacity $|E(G)|$ in $F$ such that $G_f$ is acyclic. Moreover,
  it can be determined in polynomial time.
\end{lemma}

\begin{proof}
  From Lemma \ref{flow} we know that there exists at least one flow $f$ between $s$ and $t$ of capacity $|E(G)|$.
  Let us take $f$ such that number of edges in $G_f$ is as small as possible.
  Suppose there is a cycle in $G_f$ on vertices $c_1, c_2, \ldots, c_k$ respectively.
  Denote $c_{k+1} := c_1$ as we are dealing with a cycle.
  Let $x$ be minimum of amounts of flow that $f$ sends through some edge $\ora{v_{u_iu_{i+1}}u_{i+1}}$ for some valid $i$.
  From definition of $G_f$ we deduce that $x > \frac{1}{2}$.
  Let us define $f'$ by decreasing flow $f$ on edges $\ora{v_{u_iu_{i+1}}u_{i+1}}$ 
  and increasing it on edges $\ora{v_{u_iu_{i+1}}u_{i}}$ by $x - \frac{1}{2}$.
  Amount of flow leaving and entering each vertex remains unchanged hence $f'$ still restricts capacities.
  Flow $f'$ for at least one vertex $v_{u_iu_{i+1}}$ sends exactly $\frac{1}{2}$ unit of flow through both edges outgoing from it,
  so at least one edge on the cycle is no longer present in $G_{f'}$
  and edges outside the cycle remain unchanged when compared to $G_f$.
  It contradicts the assumption that $G_f$ has the smallest possible number of edges
  what concludes proof of existence of such $f$.
  
  In~order to compute such $f$ in polynomial time let us take any $f$ of capacity
  $|E(G)|$ in $F$ (let us remind that we can determine value of $\mad(G)$ in the polynomial time
  which is used in construction of $F$). If $G_f$ contains some cycle we can detect this cycle,
  determine the value of $x$ and adjust values of units that flow sends through mentioned edges accordingly.
  Number of edges in $G_{f'}$ is strictly smaller than in $G_f$, so we will not do this
  more than $|E(G)|$ times, what gives us algorithm performing a polynomial number of operations.
  In~order to omit dealing with rational numbers we can multiply all capacities in $F$
  by $2b$, where $\mad(G) = \frac{a}{b}$ for some coprime integers $a, b$.
  That concludes description of polynomial algorithm determining the desired~$f$.

\end{proof}

Let us fix $f$ from the above lemma.
We will present an algorithm in which:
\begin{itemize}
\item the routine $NoInEdges(H_f)$ returns any vertex from directed acyclic graph $H_f$ which has no incoming edges
(as the graph is acyclic there always exists at least one such vertex)
\item the routine $KNeighborhood(H, S, k)$ given graph $H$, subset of its vertices $S$ and integer $k$
returns set of all vertices from $H$ outside of $S$ adjacent to at least $k$ vertices from $S$.
\end{itemize}

\begin{algorithm}
\begin{algorithmic}

\caption{} \label{alg}
\Function {Solve} {$H, H_f, k$}
\State $S \gets \emptyset$
\While {$H_f \neq \emptyset$}
  \State $x \gets NoInEdges(H_f)$
  \State $S \gets S \cup \{ x \}$
  \State $H_f \gets H_f - \{ x \} - KNeighborhood(H, S, k)$
\EndWhile
\EndFunction
\Return $S$

\end{algorithmic}
\end{algorithm}

\begin{theorem} \label{algo}
  Function $Solve(G, G_f, k)$ returns $(k - 1)$-degenerate set $S$ such that $\mad(G-S) \le {\mad(G) - k}$.
\end{theorem}
\begin{proof}
  First we argue the algorithm will return $(k - 1)$-degenerate set.
  In each iteration $x$ picked by algorithm is adjacent to at most $k - 1$ already picked vertices.
  So $S$ is $(k-1)$-degenerate indeed.
  
  To show that $mad(G - S) \le mad(G) - k$ we just have to find a flow $f'$ in graph 
  $F' := F(G - S,\frac{mad(G)}{2} - \frac{k}{2})$ of capacity $E(G - S)$ thanks to Lemma \ref{flow}.
  Observe that $F'$ is a subgraph of $F$
  with capacities of edges on the third layer reduced by $\frac{k}{2}$.
  Flow $f'$ has to saturate all edges from the first layer in order to have capacity $E(G - S)$.
  On the second layer we define $f'$ using $f$,
  for each edge from the second layer of $F'$ flow $f'$ will send exactly the same amount of flow 
  as $f$ on corresponding edge in $F$. 
  Now we just have to argue that amount of flow sent by $f'$ to any node between second and third layer in $F'$ is bounded 
  by $\frac{mad(G)}{2} - \frac{k}{2}$ i.e.\ capacity of edge going from that node to sink.
  Each such node corresponds to vertex from $G - S$,
  so let us take arbitrary vertex $u \in G - S$.
  During execution of the algorithm vertex $u$ has been removed from $H_f$ as incident to some $k$ vertices already picked to $S$.
  Denote them $x_1, \ldots, x_k$ and let us consider arbitrary $x_i$.
  When the algorithm picked $x_i$ from $H_f$, there were no incoming edges to $x_i$.
  In particular in $H_f$ there was no edge $\ora{ux_i}$.
  At that time $u$ still belonged to $H_f$, so there was no edge $\ora{ux_i}$ even in $G_f$.
  Since $u$ and $x_i$ are adjacent in $G$, there was either an edge $\ora{x_iu}$ in $G_f$
  which means that flow $f$ sends more than $\frac{1}{2}$ unit of flow from $v_{ux_i}$ to $u$ in $F$
  or there was no $\ora{x_iu}$ and $\ora{ux_i}$ which means that flow $f$ sends exactly
  $\frac12$ unit of flow from $v_{ux_i}$ to $u$ in $F$.
  Through node $u$ in $F$ flow $f$ sends at most $\frac{mad(G)}{2}$ units of flow
  and for every $1 \le i \le k$ at least $\frac{1}{2}$ unit of flow comes from $v_{ux_i}$ to $u$.
  Therefore flow going through $u$ is decreased by at least $\frac{1}{2}$ unit per each $x_i$ in $F'$
  what implies that $f'$ sends at most $\frac{mad(G)}{2} - \frac{k}{2}$ units of flow to vertex $u$ in $F'$.
  
\end{proof}

  What is more, procedure $Solve(G, G_f, k)$ can be trivially implemented in a polynomial time.
  Theorem \ref{mainres} directly follows from Theorem \ref{algo}.
  \filbreak
  \begin{samepage}
  As two notable special cases we mention following corollaries:
  \begin{theorem} \label{IndSet}
   For every undirected simple graph $G$
 there exists $I \subseteq V(G)$ such that $I$ is an independent set and $\mad(G - I) \le \mad(G) - 1$.
 Moreover such $I$ can be computed in polynomial time. 
  \end{theorem}
  
  \begin{theorem} \label{Forest}
   For every undirected simple graph $G$
 there exists $F \subseteq V(G)$ such that $G[F]$ is a forest and $\mad(G - F) \le \mad(G) - 2$.
 Moreover such $F$ can be computed in polynomial time. 
  \end{theorem}
  \end{samepage}

\section{Conclusions and open problems} \label{conclusions}

As already mentioned, our result is a positive answer for the
open problem presented in \cite{EDU}, namely problem~2 from final remarks
which in turn implies sub-exponential bound on the diameter of reconfiguration graphs of
colourings for graphs with any bounded maximum average degree. However this bound
has already been improved in \cite{Feghali} to the polynomial bound depending on value of $\mad(G)$.

Our results imply many results
for some specific classes of graphs as a direct consequence
and here we mention a few of them.

Following folklore fact will come in handy in deriving some of the consequences:
\begin{fact}\label{fakcik}
 For every planar graph $G$ we have $(\mad(G) - 2)(g(G) - 2) < 4$.
\end{fact}

Based on Theorem \ref{IndSet}, we are able to improve Theorem 1 from \cite{Dross} and one of its consequences.

\begin{theorem}[\cite{Dross}]
Let $M$ be a real number such that $M <3$.  Let $d>0$ be an integer and let $G$ be a graph with $\mad(G)< M$.  If $d>\frac{2}{3-M}-2$, then $V(G)$ can be partitioned into $A \uplus B$ such that $G[A]$ is an independent set and $G[B]$ is a forest with maximum degree at most $d$.
\end{theorem}

By using Theorem \ref{IndSet} and direct calculation, we are able to strengthen this to the following theorem:

\begin{theorem} \label{DrossRekt} Let $M$ be a real number such that $M <3$.  Let $d>0$ be an integer and let $G$ be a graph with $\mad(G)< M$.  If $d>\frac{2}{3-M}-2$, then $V(G)$ can be partitioned into $A \uplus B$ such that $G[A]$ is an independent set and $G[B]$ is a forest whose connected components have size at most $d + 1$.
\end{theorem}

Authors of \cite{Dross} use their theorem and Fact \ref{fakcik} to deduce the following corollary:
\begin{corollary}
 For every planar graph $G$ with $g(G) \ge 10$, the vertex set $V(G)$ can be partitioned into $A \uplus B$ such that $A$ is an independent set and $G[B]$ is a forest with maximum degree $2$.
\end{corollary}

We are able to strengthen this to the following corollary:
\begin{corollary}
 For every planar graph $G$ with $g(G) \ge 10$, the vertex set $V(G)$ can be partitioned into $A \uplus B$ such that $A$ is an independent set and
 $G[B]$ is a forest whose connected components have size at most $3$.
\end{corollary}

Borodin et al. proved in \cite{YETANOTHERBORODIN} that vertex set of any planar graph with $g(G) \ge 7$ admits a $(\mathcal{I}, \Delta_4)$-partition, i.e. partition into an independent set and a set that induces graph with maximum degree at most $4$. Dross et al. proved in \cite{Dross} that vertex set of any planar graph with $g(G) \ge 7$ admits a $(\mathcal{I}, \mathcal{F}_3)$-partition, i.e. into an independent set and a set that induces forest of max degree at most $5$. In Corollary \ref{planar7} we add another partition result for the class of planar graphs with girth at least $7$.

\begin{corollary} \label{planar7}
 For every planar graph $G$ with $g(G) \ge 7$, the vertex set $V(G)$ can be partitioned
 into $A \uplus B$ such that $A$ is an independent set and $G[B]$ is a forest
 where every connected component has at most~$9$ vertices.
\end{corollary}
\begin{proof}
 Since $g(G) \ge 7$ we deduce that $\mad(G) < 1 + \frac95$, so based on Corollary \ref{IndSet}
 we get that there exist $A$ and $B$ such that $V(G) = A \uplus B$, $A$ is an independent set
 and $\mad(G[B]) < \frac95$.
 It can be readily verified that class of graphs with value of their $\mad$ smaller than $\frac95$ is class of graphs which are forests
 with connected components of size at most $9$.
\end{proof}

Apart from that, based on Theorems \ref{IndSet} and \ref{Forest} and Fact \ref{fakcik} we are able to deduce following corollaries:

\begin{corollary} \label{3lasy}
 For every planar graph $G$, the vertex set $V(G)$ can be partitioned
 into $A \uplus B \uplus C$ such that $G[A], G[B], G[C]$ are forests.
\end{corollary}
\begin{proof}
 Every planar graph satisfies $\mad(G) < 6$, so using Corollary \ref{Forest}
 we can partition $V(G)$ into $A$ and $D$ such that $\mad(G[A]) < 2$ and $\mad(G[D]) < 4$
 and then using Corollary \ref{Forest}
 again we can partition $D$ into $B$ and $C$ such that $\mad(G[B]) < 2$ and $\mad(G[C]) < 2$.
 Hence $G[A], G[B], G[C]$ are forests.
\end{proof}

\begin{corollary} \label{2lasy}
 For every planar graph $G$ without triangles, the vertex set $V(G)$ can be partitioned
 into $A \uplus B$ such that $G[A], G[B]$ are forests.
\end{corollary}
\begin{proof}
 Based on Fact \ref{fakcik} we know that if $G$ has no triangles then $g(G) \ge 4 \Rightarrow g(G) - 2 \ge 2
 \Rightarrow \mad(G) < 4$. Therefore using Corollary \ref{Forest}
 we deduce that there exist $A, B$ such that $V(G) = A \uplus B$ and $G[A], G[B]$ are forests.
\end{proof}

\begin{corollary}\label{slabe1}
 For every planar graph $G$ without cycles of length $3$ and $4$, the vertex set $V(G)$
 can be partitioned into $A \uplus B$ such that $G[A]$ is a forest and $\Delta(G[B]) \le 1$.
\end{corollary}
\begin{proof}
 Since $g(G) \ge 5$ we deduce that $\mad(G) < 2 + \frac{4}{3}$, so based on Corollary \ref{Forest}
 we get that there exist $A$ and $B$ such that $V(G) = A \uplus B$ and $\mad(G[A]) < 2$
 and $\mad(G[B]) < \frac43$. Therefore $G[A]$ is a forest and $\Delta(G[B]) \le 1$, because
 if $G[B]$ contains a vertex with degree $\ge 2$ then this vertex together with its two neighbours
 induce a graph with $\mad$ at least $\frac43$.
\end{proof}

\begin{corollary}\label{slabe2}
 For every planar graph $G$ with $g(G) \ge 6$ its vertex set $V(G)$ can be partitioned
 into $A \uplus B$ such that $G[A]$ is a forest and $B$ is an independent set.
\end{corollary}
\begin{proof}
 Since $g(G) \ge 6$ we deduce that $\mad(G) < 3$, so based on either Corollary \ref{Forest}
 or Corollary \ref{IndSet}
 we get that there exist $A$ and $B$ such that $V(G) = A \uplus B$ and $\mad(G[A]) < 2$
 and $\mad(G[B]) < 1$. Therefore $G[A]$ is a forest and $B$ is an independent set.
\end{proof}

However, Corollaries \ref{3lasy}, \ref{2lasy}, \ref{slabe1} and \ref{slabe2} have already been proven and even improved before.
Corollary \ref{3lasy} has been proven in \cite{Chartrand1968} and later improved in \cite{LinArbo}.
Improved version of Corollary \ref{2lasy} has been proven in \cite{VertexArbo}.
Improved version of both Corollaries \ref{slabe1} and \ref{slabe2} has been proven in \cite{Girth5}.

As a main open problem in this area we recall the following one:
\begin{conjecture}
 For every graph $G$ and positive real numbers $c_1, c_2$ if $\mad(G) < c_1 + c_2$ then
 there exists a~partition of vertex set $V(G) = A \uplus B$ such that $\mad(G[A]) < c_1$
 and $\mad(G[B]) < c_2$.
\end{conjecture}

Our results show that this conjecture is true for $c_2 \in \{1, 2\}$.
Moreover since for positive $k$ we have that $k$-degenerate graphs fulfill $\mad(G) < 2k$
we can deduce that for every integer $k \ge 2$ and a graph $G$ that satisfies $\mad(G) < c_1 + k$
there exists a partition of vertex set $V(G) = A \uplus B$ such that
$\mad(G[A]) < c_1$ and $\mad(G[B]) < 2k-2$.

\section*{Acknowledgements}
Majority of this research has been conducted during Structural Graph Theory workshop in Gu\l{}towy, Poland, in June~2019. Thanks to the organizers and to the other workshop participants for creating a productive working atmosphere.
We would also like to thank Bartosz Walczak for discussions on this topic.



\bibliography{MAD}{}
\bibliographystyle{plain}

\end{document}